\documentclass[11pt,a4paper]{article}

\usepackage[a4paper,text={150mm,240mm},centering,headsep=10mm,footskip=15mm]{geometry}
\usepackage[T1]{fontenc}
\usepackage{graphicx}
\usepackage[rightcaption]{sidecap}
\usepackage{epsfig, float}
\usepackage{pgf,tikz,pgfplots}
\usetikzlibrary{arrows}
\usepackage{enumitem}
\usepackage{dcolumn}
\pgfplotsset{compat=1.10}
\usepackage{amsmath,amssymb}
\usepackage{amsfonts}
\usepackage{url}
\usepackage{bbm}
\usepackage{mathrsfs} 
\usepackage[english]{babel}
\usepackage[unicode=true, pdfusetitle, bookmarks=true,
  bookmarksnumbered=false, bookmarksopen=false, breaklinks=true, 
  pdfborder={0 0 0}, backref=false, colorlinks=true, linkcolor=blue,
  citecolor=blue, urlcolor=blue]{hyperref}
\usepackage{slashed}
\usepackage{dsfont}
\usepackage{mathabx}
\usepackage{cite}
\usepackage[normalem]{ulem}

\newcommand{\R}{\mathbb{R}}
\newcommand{\CC}{\mathbb{C}}
\newcommand{\N}{\mathbb{N}}

\newcommand{\M}{\mathbb{M}}
\newcommand{\id}{\mathbbm{1}}

\newcommand{\vx}{\mathbf{x}}
\newcommand{\vy}{\mathbf{y}}

\newcommand{\be}{\begin{equation}}
\newcommand{\ee}{\end{equation}}
\newcommand{\ret}{{\rm ret}}
\newcommand{\free}{{\rm free}}
\newcommand{\Dirac}{{\rm Dirac}}
\newcommand{\vertiii}[1]{{\left\vert\kern-0.25ex\left\vert\kern-0.25ex\left\vert #1 
    \right\vert\kern-0.25ex\right\vert\kern-0.25ex\right\vert}}
\newcommand{\Banach}{\mathscr{B}}

\DeclareMathOperator*{\esssup}{ess \, sup}

\newtheorem{theorem}{Theorem}[section]
\newtheorem{lemma}[theorem]{Lemma}

\newenvironment{proof}[1][Proof:]{\begin{trivlist}
\item[\hskip \labelsep {\bfseries #1}]}{\end{trivlist}}

\newcommand{\qed}{\hfill\ensuremath{\Box}}

\title{Existence of relativistic dynamics for two directly interacting Dirac particles in 1+3 dimensions}

\author{
	Matthias Lienert\thanks{Institute Cyber Defense (CODE),
     Bundeswehr University Munich,
     Carl-Wery-Str. 22, 81739 Munich.
     E-mail: matthias.lienert@unibw.de} \ and
Markus N\"oth\thanks{Mathematisches Institut, Ludwig-Maximilians-Universit\"at,
	Theresienstr. 39, 80333 M\"unchen, Germany. E-mail: noeth@math.lmu.de}
}

\date{March 18, 2021}

\begin{document}

\maketitle

\begin{abstract}
\noindent  Here we prove the existence and uniqueness of solutions of a class of integral equations describing two Dirac particles in 1+3 dimensions with direct interactions. This class of integral equations arises naturally as a relativistic generalization of the integral version of the two-particle Schr\"odinger equation. Crucial use of a multi-time wave function $\psi(x_1,x_2)$ with $x_1,x_2 \in \R^4$ is made. A central feature is the time delay of the interaction. Our main result is an existence and uniqueness theorem for a Minkowski half space, meaning that Minkowski spacetime is cut off before $t=0$. We furthermore show that the solutions are determined by Cauchy data at the initial time; however, no Cauchy problem is admissible at other times. A second result is to extend the first one to particular FLRW spacetimes with a Big Bang singularity, using the conformal invariance of the Dirac equation in the massless case. This shows that the cutoff at $t=0$ can arise naturally and be fully compatible with relativity. We thus obtain a class of interacting, manifestly covariant and rigorous models in 1+3 dimensions.
\\[0.1cm]

\end{abstract}

\section{Introduction}

The Dirac equation is perhaps the most important equation in relativistic quantum theory, thus it may seem surprising that no completely satisfactory mathematical mechanism of interaction has been found for it. Usually, interactions between many particles are implemented in one of the following ways: (a) adding a potential to the free Hamiltonian, (b) using a second quantized electromagnetic field which mediates the interaction. Both approaches face difficulties. Approach (a) corresponds to postulating the equation
\be
	i \partial_t \varphi(t,\vx_1,\vx_2) = \left(H_1^\Dirac + H_2^\Dirac + V(t,\vx_1,\vx_2) \right) \varphi(t,\vx_1,\vx_2),
\label{eq:singletimedirac}
\ee
where $V$ is a potential and $H_k^\Dirac$ the Dirac Hamiltonian acting on the variables of the $k$-th particle. Under appropriate circumstances, it is clear that \eqref{eq:singletimedirac} defines an interacting dynamics  (see e.g.\ \cite{dirk_martin_2018} and references therein). 
However, \eqref{eq:singletimedirac} is not Lorentz invariant.

Approach (b), on the other hand, easily leads to a Lorentz invariant dynamics. However, one encounters difficulties with ultraviolet divergences. These difficulties have led to the situation that, great efforts notwithstanding, it has so far only been possible to rigorously define a Lorentz invariant dynamics for toy models in 1+1 and 1+2 spacetime dimensions (see e.g.\ \cite{thirring_model,glimm_jaffe,jaffe_cft}). In 1+3 dimensions, it has been an open problem to prove the existence of the dynamics for any interacting and completely relativistic model.

In this paper, we pursue a new approach to defining interacting dynamics, neither via potentials nor via second quantized fields, but rather through \textit{direct interactions with time delay}, and prove the existence of dynamics for the simple case of two Dirac particles in 1+3 dimensions. The key innovation is to make use of \textit{multi-time wave functions}. This concept goes back to Dirac \cite{dirac_32}, played an important role in the works of Tomonaga \cite{tomonaga} and Schwinger \cite{schwinger}, has been studied by different authors over the years \cite{guenther_1952,marx_1974,schweber,drozvincent_1981,sazdjian_2bd,2bdem} and has recently undergone considerable developments \cite{nogo_potentials,qftmultitime,multitime_pair_creation,1d_model,nt_model,2bd_current_cons,deckert_nickel_2016,lpt_2017b,generalized_born,ibc_model,phd_nickel}; an overview can be found in \cite{dice_paper}. For two Dirac particles in Minkowski spacetime $\M$, a multi-time wave function is a map
\be
	\psi : \M \times \M \rightarrow \CC^4 \otimes \CC^4\cong \CC^{16},~~~(x_1,x_2) \mapsto \psi(x_1,x_2).
\label{eq:multitimewavefn}
\ee	
$\psi$ can be considered a generalization of the single-time wave function $\varphi$ in the Schr\"odinger picture, as in Eq.\ \eqref{eq:singletimedirac}. The relation of $\psi$ to $\varphi$ is straightforwardly given by
\be
	\varphi(t,\vx_1,\vx_2) = \psi((t,\vx_1),(t,\vx_2)).
	\label{eq:singlemulti}
\ee
Contrary to the single-time wave function $\varphi$ (which refers to a frame), $\psi$ is a manifestly covariant object. Under a Poincar\'{e} transformation $(a,\Lambda)$, $\psi$ transforms as
\be
	\psi'(x_1,x_2) = S[\Lambda]\otimes S[\Lambda] \psi(\Lambda^{-1}(x_1-a), \Lambda^{-1}(x_2-a)),
\ee
where $S[\Lambda]$ are the matrices appearing in the spinor representation of the Lorentz group.

For the present purposes, it is crucial that $\psi$ is defined on general space-time configurations $(x_1,x_2)\in \M\times \M$, not only on equal-time configurations as $\varphi$. By relating configurations $(x_1,x_2)$ with different time coordinates $x_1^0 \neq x_2^0$ one can express \textit{interactions with a time delay}. It has been pointed out in \cite{direct_interaction_quantum} that in this way, \textit{direct relativistic interactions} (unmediated by fields) can be expressed at the quantum level. 
In particular, it becomes possible to formulate a quantum analog of direct interactions along light cones, such as in the Wheeler-Feynman formulation of classical electrodynamics \cite{wf1,wf2}, using values of $\psi(x_1,x_2)$ with $(x_1-x_2)_\mu (x_1-x_2)^\mu = 0$. This is not directly feasible using just $\varphi$. We thus note that \textit{new kinds of interacting quantum dynamics can be defined using a multi-time wave function}.

An interesting class of such dynamics has recently been suggested in \cite{direct_interaction_quantum} and has been subsequently analyzed rigorously in \cite{mtve,int_eq_curved}: \textit{multi-time integral equations}. 
But why study integral equations instead of PDEs?  To answer this question, note that the initial value problem $\varphi(0,\vx_1,\vx_2) = \psi_0(\vx_1,\vx_2)$ of the single-time Schr\"odinger equation \eqref{eq:singletimedirac} can equivalently be formulated as the following integral equation:
\begin{align}
	\varphi(t,\vx_1,\vx_2) = \varphi^\free(t,\vx_1,\vx_2) + \int_0^\infty \!\! dt' \int d^3 \vx_1' \, d^3 \vx_2' &\, \gamma_1^0 S_1^\ret(t-t',\vx_1-\vx_1') \gamma_2^0 S_2^\ret(t-t',\vx_2-\vx_2') \nonumber\\
&\times V(t',\vx_1',\vx_2')\varphi(t',\vx_1',\vx_2'),
	\label{eq:singletimediracint}
\end{align}
where $\varphi^\free$ is the solution of the same initial value problem of the free equation (\eqref{eq:singletimedirac} with $V=0$) and $S_k^\ret$ is the retarded Green's function of the $k$-th Dirac operator.

Now, contrary to the PDE \eqref{eq:singletimedirac}, the integral equation \eqref{eq:singletimediracint} possesses a straightforward manifestly covariant generalization in terms of a multi-time wave function, namely:
\be
	\psi(x_1,x_2) = \psi^\free(x_1,x_2) + \int d^4 x_1' \, d^4 x_2' \, S_1(x_1-x_2') S_2(x_2-x_2') K(x_1',x_2') \psi(x_1',x_2'),
	\label{eq:inteq}
\ee
where $\psi^\free$ is a solution of the equations $D_1 \psi^\free = 0$, $D_2\psi^\free = 0$, $D_k = (i \gamma^\mu_k \partial_{k,\mu} - m_k)$ and $S_1, S_2$ are (retarded or other) Green's functions of $D_1, D_2$, respectively. $K(x_1,x_2)$ denotes the so-called \textit{interaction kernel}, a Poincar\'{e} invariant function (or distribution) which generalizes the potential in Eq.\ \eqref{eq:singletimediracint}. The crucial point is that \eqref{eq:inteq} incorporates interactions with time delay which cannot be expressed through a PDE. It has been demonstrated in \cite{direct_interaction_quantum} that for $K(x_1,x_2) \propto \delta((x_1-x_2)_\mu (x_1-x_2)^\mu)$, the Dirac delta distribution along the light cone, one re-obtains \eqref{eq:singletimedirac} with $V(t,\vx_1,\vx_2) \propto \frac{1}{|\vx_1-\vx_2|}$ if one neglects the time delay of the interaction. Thus, \eqref{eq:inteq} constitutes a natural generalization of \eqref{eq:singletimediracint}.

Further support for considering the integral equation \eqref{eq:inteq} comes from the fact that the Bethe-Salpeter (BS) equation of QFT \cite{bs_equation}, which is usually considered an effective equation for a bound state, has a similar form as \eqref{eq:inteq}. That being said, there are also significant physical and mathematical differences between the two equations (see \cite[sec. 3.3]{direct_interaction_quantum}).

\paragraph{Previous results.} To the best of our knowledge, the first results about the existence and uniqueness of dynamics for Eq.\ \eqref{eq:inteq} have been obtained in \cite{mtve}, for the case of a Minkowski half-space and Klein-Gordon (KG) particles. A "Minkowski half-space" means to use $\frac{1}{2} \M \times \frac{1}{2}\M$ with $\frac{1}{2} \M = [0,\infty) \times \R^3$, i.e.\, Minkowski spacetime cut off before $t=0$, as the domain of integration in \eqref{eq:inteq}. The KG case refers to replacing $S_1, S_2$ with (retarded) Green's functions of the KG equation and $\psi^\free$ with a solution of $(\Box_k +m^2_k)\psi^\free = 0,~k=1,2$. The main result in \cite{mtve} was to show that for every $\psi^\free$ which is $L^2$ in the spatial directions and $L^\infty$ in the time directions there is a unique solution $\psi$ with the same properties. In addition, at $t_1=t_2=0$, $\psi^\free$ and $\psi$ agree so that one actually has a Cauchy problem at the initial time. In order to obtain that result, the interaction kernel was assumed to be either bounded or to just have a $1/|\vx_1-\vx_2|$ singularity. In 1+3 dimensions, only the massless case was treated. The proof was based on exploiting a Volterra property which appears for retarded Green's functions and $\frac{1}{2}\M$, i.e.\, the time integrations in \eqref{eq:inteq} reach only from 0 to $x_1^0 $ or $x_2^0$ (given by the time arguments of $\psi$ on the left hand side). This allowed an effective iteration scheme for Eq.\ \eqref{eq:inteq}, leading to a global existence and uniqueness result for a formidable-looking non-Markovian (history dependent) type of dynamics.

The cutoff of spacetime at $t=0$ was introduced in \cite{mtve} to obtain the Volterra property. While such a cutoff destroys Lorentz invariance, there could be physical justification for a beginning in time which is compatible with relativity. Such a justification has been provided in \cite{int_eq_curved}. There, the integral equation was extended to curved spacetimes and analyzed in more detail for certain spacetimes which feature a Big Bang singularity, Friedman-Lema\^itre-Robertson-Walker (FLRW) spacetimes. The Big Bang then provides a natural cutoff in the cosmological time. In this way, the existence of certain classes of fully covariant dynamics for massless KG particles was demonstrated.

\paragraph{Goal of the paper.} Here we would like to extend the previous results to the case of Dirac instead of KG particles. This is desirable as the Dirac equation describes actual elementary particles (fermions) while the KG equation is usually considered only a toy equation as its currents do not have the right properties to play the role of a probability current. Mathematically, the Dirac case is more challenging than the KG case as contrary to the latter, the Dirac Green's functions contain distributional derivatives. A Green's function of the Dirac equation is given by acting with the adjoint Dirac operator $\overline{D} = (-i \gamma^\mu \partial_\mu - m)$ on a Green's function $G(x)$ of the KG equation, i.e.\:
\be
	S(x)= \overline{D}G(x).
\ee
Consequently, one has to define the integral operator in \eqref{eq:inteq} on a function space where one can take certain weak derivatives. In contrast to most of non-relativistic physics, this also concerns the time derivatives here. The choice of function space can be a tricky issue, as the convergence of an iteration scheme (and of the Neumann series, our strategy of proof) requires the integral operator to preserve the regularity, so that the regularity needs to be in harmony with the structure of the integral equation (see Sec.\ \ref{sec:choiceofB}).

\paragraph{Further motivation.}
\begin{enumerate}
	\item It is quite challenging to set up an interacting dynamics for multi-time wave functions. The issue here is not only Lorentz invariance but rather the mere compatibility of the time evolutions in the various time coordinates. A no-go theorem \cite{nogo_potentials,deckert_nickel_2016} for example rules out interaction potentials (which could be Poincar\'{e} invariant functions in the multi-time approach). Thus, interaction is more difficult to achieve for multi-time than for single-time wave functions. So far, the only rigorous, interacting and Lorentz invariant multi-time models for Dirac particles have been constructed in 1+1 spacetime dimensions \cite{1d_model,nt_model} (see, however, \cite{drozvincent_1981,sazdjian_2bd,2bdem} for non-rigorous Lorentz invariant models in 1+3 dimensions and \cite[chap.\ 3]{phd_nickel} for a not fully Lorentz invariant but rigorous model in 1+3 dimensions). Considering these difficulties, the multi-time aspect of our model is interesting in its own right.
	\item Eq.\ \eqref{eq:inteq} defines, in the case of retarded Green's functions, a new class of Volterra-type equations which may be interesting also for researchers specializing in integral equations. It provides a reason why a multi-dimensional Volterra-type equation would be relevant for physics, and shows which properties to expect for applications.
\end{enumerate}

\paragraph{Overview.} The paper is structured as follows. In Sec.\ \ref{sec:setting}, we specify the integral equation \eqref{eq:inteq} in detail. The difficulties with understanding the distributional derivatives are discussed and a suitable function space is identified. Sec.\ \ref{sec:results} contains our main results. In Sec.\ \ref{sec:minkhalfspace}, we formulate an existence and uniqueness theorem (Thm.\ \ref{thm:minkhalfspace}) for Eq.\ \eqref{eq:inteq} on $\frac{1}{2}\M$. It is shown that the relevant initial data are equivalent to Cauchy data at $t=0$. In Sec.\ \ref{sec:flrw}, we provide a physical justification for the cutoff at $t=0$ by extending the results to a FLRW spacetime. In the massless case, we show that an existence and uniqueness theorem can be obtained from the one for $\frac{1}{2}\M$ via conformal invariance. The result, Thm.\ \ref{thm:flrw}, covers a fully relativistic interacting dynamics in 1+3 spacetime dimensions. The proofs are carried out in Sec.\ \ref{sec:proofs}. Sec.\ \ref{sec:discussion} contains a discussion and an outlook on future research.

\section{Setting of the problem} \label{sec:setting}

\subsection{Definition of the integral operator on test functions} \label{sec:aontestfunctions}

In this section, we show how the integral operator in \eqref{eq:inteq} can be defined rigorously on test functions.
We consider the integral equation \eqref{eq:inteq} on the Minkowski half space $\frac{1}{2}\M := [0,\infty) \times \R^3$ equipped with the metric $g = \text{diag}(1,-1,-1,-1)$. We focus on retarded Green's functions of the Dirac equation, $S^\ret(x) = \overline{D} G^\ret(x)$ where $G^\ret(x)$ is the retarded Green's function of the KG equation. Explicitly,
\be
	G^\ret(x) = \frac{1}{4\pi} \frac{\delta(x^0-|\vx|)}{|\vx|} - \frac{m}{4\pi} H(x^0-|\vx|) \frac{J_1(m \sqrt{x^2})}{\sqrt{x^2}}
\ee
where $H$ denotes the Heaviside function, $J_1$ the Bessel function of first kind of order 1 and $x^2 = (x^0)^2 - |\vx|^2$.

In order to define the meaning of the Green's functions as distributions, we introduce a suitable space of test functions:
\be
	\mathscr{D} = \mathscr{S}\big( (\tfrac{1}{2}\M)^2,\CC^{16} \big),
\ee
the space of 16-component Schwartz functions on $(\tfrac{1}{2}\M)^2$.
For a smooth interaction kernel $K$ and a test function $\psi \in \mathscr{D}$, we then understand \eqref{eq:inteq} by formally integrating by parts so that all partial derivatives act on $K \psi$:
\begin{align}
	\psi(x_1,x_2) =\, &\psi^\free(x_1,x_2) +\int_{\frac{1}{2}\M} d^4 x_1' \int_{\frac{1}{2}\M} d^4 x_2' \, G_1^\ret(x_1-x_1') G_2^\ret(x_2-x_2') [D_1 D_2(K \psi)](x_1',x_2')\nonumber\\
&+\text{boundary terms},
\label{eq:inteqpi}
\end{align}
where $D_k = (i \gamma_k^\mu \partial_{x_k^\mu} - m_k)$, $k=1,2$. The boundary terms result from the fact that $\psi(x_1,x_2) \neq 0$ for $x_1^0 = 0$ or $x_2^0 =0$ and are given by:
\begin{align}
	&\int_{\R^3} d^3 \vx_1' \int_{\R^3} d^3 \vx_2' \left. i \gamma_1^0 G_1^\ret(x_1-x_1') i \gamma_2^0 G_2^\ret(x_2-x_2') (K \psi)(x_1',x_2') \right|_{{x_1^0}' = 0,\, {x_2^0}' = 0}\nonumber\\
+ &\int_{\R^3} d^3 \vx_1' \int_{\frac{1}{2} \M} d^4 x_2' \left. i \gamma_1^0 G_1^\ret(x_1-x_1') G_2^\ret(x_2-x_2') D_2 (K \psi)(x_1',x_2') \right|_{{x_1^0}' = 0}\nonumber\\
+  &\int_{\frac{1}{2} \M} d^4 x_1' \int_{\R^3} d^3 \vx_2' \left. G_1^\ret(x_1-x_1')  i \gamma_2^0 G_2^\ret(x_2-x_2') D_1 (K \psi)(x_1',x_2') \right|_{{x_2^0}' = 0}.
\label{eq:boundaryterms}
\end{align}
Now, $G_k^\ret$ still contains the $\delta$-distribution. We use the latter to cancel the integrals over ${x_k^0}'$, $k=1,2$ in \eqref{eq:inteqpi} in the following manner.
\begin{align}
	\frac{1}{4\pi} \int_{\tfrac{1}{2}\M} d^4 x' \, \frac{\delta(x^0-{x^0}'-|\vx-\vx'|)}{|\vx-\vx'|} f(x')
= ~&\frac{1}{4\pi} \int_{B_{x^0}(\vx)} d^3 \vx' \, \frac{1}{|\vx-\vx'|} f(x')|_{{x^0}'=x^0-|\vx-\vx'|}\nonumber\\
= ~&\frac{1}{4\pi} \int_{B_{x^0}(0)} d^3 \vy\, \frac{1}{|\vy|} f(x +y)|_{y^0=-|\vy|}.
\end{align}
Moreover,
\begin{align}
	&\frac{m}{4\pi} \int_{\tfrac{1}{2}\M} \!\! d^4 x' ~ H(x^0 - {x^0}' - |\vx-\vx'|) \frac{J_1(m\sqrt{(x-x')^2})}{\sqrt{(x-x')^2}} f(x')\nonumber\\
= ~& \frac{m}{4\pi} \int_{[-x^0,\infty) \times \R^3} \!\!\!\! d^4 y ~ H(-y^0- |\vy|) \frac{J_1(m\sqrt{y^2})}{\sqrt{y^2}} f(x +y)\nonumber\\
= ~& \frac{m}{4\pi} \int_{-x^0}^0 dy^0 \int_{B_{|y^0|}(0)} d^3 \vy_k \, \frac{J_1(m\sqrt{y^2})}{\sqrt{y^2}} f(x+y).
\end{align}
For the boundary terms, we similarly use
\be
	 \frac{i \gamma^0}{4\pi} \int_{\R^3}d^3 \vx' ~ \frac{\delta(x^0-|\vx-\vx'|)}{|\vx-\vx'|}f(0,\vx') ~=~ \frac{i \gamma^0}{4\pi} \int_{\partial B_{x^0}(0)}d\sigma(\vy) ~ \frac{f(0,\vx+\vy)}{x^0}
\ee
as well as
\begin{align}
	&i \gamma^0\frac{m}{4\pi} \int_{\R^3}d^3 \vx' ~ H(x^0 - {x^0}' - |\vx-\vx'|) \frac{J_1(m\sqrt{(x-x')^2})}{\sqrt{(x-x')^2}} f(x')|_{{x^0}' = 0}\nonumber\\
&= i \gamma^0\frac{m}{4\pi} \int_{B_{x^0}(0)}d^3 \vy ~ \frac{J_1(m\sqrt{(x^0)^2 - \vy^2})}{\sqrt{(x^0)^2 - \vy^2}} f(0,\vx+\vy).
\end{align}
This yields the form of the integral equation which shall be the basis of our investigation:
\be
	\psi(x_1,x_2) = \psi^\free(x_1,x_2) + (A \psi)(x_1,x_2).
\label{eq:inteqschematic}
\ee
The operator $A$ is first defined on test functions $\psi \in \mathscr{D}$ as
\be
	A\psi ~=~ \prod_{j=1,2} \left(A_j^{(1)}(m)D_j + A_j^{(2)}(m)D_j + A_j^{(3)}(m) + A_j^{(4)}(m)\right) K\psi
 \label{eq:defa}   
\ee
where for $j=1,2$, $k=1,2,3,4$ the operator $A_j^{(k)}(m) : \mathscr{D} \rightarrow C^\infty\big((\frac{1}{2}\M)^2,\CC^{16}\big)$  is defined by letting the respective operator $A^{(k)}(m)$ on $\mathscr{S}\big( \tfrac{1}{2} \M, \CC^4 \big)$ , given below, act on the $j$-th 4-variable and spin index of $\psi(x_1,x_2),~ \psi \in \mathscr{D}$.\footnote{We deliberately avoid using tensor products here, as the tensor product of Banach spaces is an ambiguous notion.}
\begin{align}
	\left(A^{(1)}(m) \,f \right)(x) ~&=~ \frac{1}{4\pi} \int_{B_{x^0}(0)} \! \!d^3 \vy ~ \frac{1}{|\vy|} f(x+y)|_{y^0 = -|\vy|}, \label{eq:a1}\\
	\left(A^{(2)}(m) \,f\right)(x) ~&=~ -\frac{m}{4\pi} \int_{-x^0}^0 dy^0 \int_{B_{|y^0|}(0)} d^3 \vy ~ \frac{J_1(m\sqrt{y^2})}{\sqrt{y^2}} f(x+y),\label{eq:a2}\\
	\left(A^{(3)}(m)\, f\right)(x) ~&=~ \frac{i \gamma^0}{4\pi} \int_{\partial B_{x^0}(0)}d\sigma(\vy) ~ \frac{f(0,\vx+\vy)}{x^0}, \label{eq:a3}\\
	\left(A^{(4)}(m) \, f\right)(x) ~&=~ - i \gamma^0 \frac{m}{4\pi} \int_{B_{x^0}(0)}d^3 \vy ~ \frac{J_1(m\sqrt{(x^0)^2 - \vy^2})}{\sqrt{(x^0)^2 - \vy^2}} f(0,\vx+\vy).\label{eq:a4}
\end{align}
Here, the dependence of $A_j^{(1)}$ and $A_j^{(3)}$ on $m$ is only for notational convenience.\\
We now turn to the question of a suitable Banach space for Eq.\ \eqref{eq:inteqschematic}.

\subsection{Choice of Banach space} \label{sec:choiceofB}

In order to prove the existence and uniqueness of solutions, we would like to demonstrate the convergence of the Neumann series. First of all, this requires to extend the integral operator $A$ to an operator on a suitable Banach space $\Banach$. The behavior of solutions $\psi^\free(x_1,x_2)$ of the free Dirac equation in each spacetime variable $x_1,x_2$ suggests to choose the Bochner space
	\be
		\Banach_0 = L^\infty \left([0,\infty)^2_{(x_1^0,x_2^0)}, \,  L^2(\R^6,\CC^{16})_{(\vx_1,\vx_2)}\right)
	\label{eq:banach0}
	\ee
with norm
\be
	\| \psi \|_{\Banach_0} = \esssup_{x_1^0,x_2^0 > 0} \, \| \psi(x_1^0,\cdot,x_2^0,\cdot)\|_{L^2}.
\ee
The reason for choosing $\Banach_0$ is that the spatial norm $\| \psi^\free(x_1^0,\cdot,x_2^0,\cdot)\|_{L^2}$ of a solution of the free Dirac equations is constant in the two time variables $x_1^0,x_2^0$. A very similar space as $\Banach_0$ has been used for analyzing \eqref{eq:inteq} in the KG case \cite{mtve}.

However, as \eqref{eq:defa} involves the Dirac operators $D_1, D_2$, $\Banach_0$ is not sufficient for our problem. An appropriate Banach space $\Banach$ must allow us to take at least weak derivatives of $\psi$.  The choice of $\Banach$ is a delicate matter. One can easily go wrong with demanding too much regularity, as we shall illustrate now.

\paragraph{Possible difficulties with the choice of function space.}

The problem can best be illustrated with an example which is structurally related to \eqref{eq:inteq} but otherwise simpler. Consider the equation for $f: \R^2 \rightarrow \CC$,
\be
	f(t,z) = f^\free(t,z) + \int_0^t dz' \, K(z,z') \partial_t f(t,z'),
	\label{eq:modelinteq}
\ee
where $f^\free, K : \R^2 \rightarrow \CC$ are given. \eqref{eq:modelinteq} is inspired by the term $A_1 D_1$ in \eqref{eq:defa}.

We would like to set up an iteration scheme for \eqref{eq:modelinteq}. As we cannot integrate by parts to shift the $t$-derivative to $K$, we must demand at least weak differentiability of $f$ with respect to $t$. This suggests using a Banach space such as $\Banach = H^1(\R^2)$.
To prove that the integral operator in \eqref{eq:modelinteq} maps $\Banach$ to $\Banach$ (the first step in every iteration scheme), we then have to estimate the $L^2$-norm of
\be
	\partial_t \int_0^t dz' \, K(z,z') \partial_t f(t,z') = K(z,t) (\partial_t f)(t,t) + \int_0^t dz' \, K(z,z') \partial_t^2 f(t,z').
\ee
This expression, however, contains $\partial_t^2 f$. For this to make sense, we must be allowed to take the second weak time derivative of $f$. This, in turn, requires to choose a different Sobolev space, such as $H^2(\R^2)$, and to estimate the $L^2$-norm of the second time derivative of the integral operator acting on $f$ which involves $\partial_t^3 f$, and so on. One is thus led to a Sobolev space where all weak $n$-th time derivatives have to exist. Such infinite-order Sobolev spaces have, in fact, been investigated in \cite{dubinskii_1991}. However, it does not seem realistic to get an iteration to converge on these spaces. We therefore take a different approach. 

\paragraph{A Banach space adapted to our integral equation.} Considering the form of the integral operator $A$ \eqref{eq:defa}, one can see that it is sufficient that the derivatives $D_1 \psi$, $D_2 \psi$ and $D_1 D_2 \psi$ exist in a weak sense. As we want to prove later that $A$ maps the Banach space to itself, we have to estimate, among other things, a suitable norm of $D_1 (A \psi)$. If $\psi \in \mathscr{D}$ is a test function and $K$ is smooth, we have
\begin{align}
	D_1 (A \psi)(x_1,x_2) &= D_1 \int d^4 x_1' \, d^4 x_2' \, S_1(x_1-x_2') S_2(x_2-x_2') K(x_1',x_2') \psi(x_1',x_2')\nonumber\\
	&= \int d^4 x_2' \, S_2(x_2-x_2') K(x_1,x_2') \psi(x_1,x_2')
\label{eq:d1apsi}
\end{align}
where we have used $D_1 S_1(x_1-x_1') = \delta^{(4)}(x_1-x_1')$. The crucial point now is that \eqref{eq:d1apsi} does not contain higher-order derivatives such as $D_1^2 \psi$. The same holds true also for $D_2 (A \psi)$ and $D_1 D_2 (A\psi)$. Thus, the problem of the toy example \eqref{eq:modelinteq} is avoided.

Together with the previous considerations about $\Banach_0$ \eqref{eq:banach0}, we are led to define the Banach space $\Banach_g$ as the completion of $\mathscr{D}$ with respect to the following Sobolev-type norm:
\be
	\| \psi \|^2_g = \esssup_{x_1^0,x_2^0 >0} \frac{1}{g(x_1^0)g(x_2^0)} [\psi]^2(x_1^0,x_2^0)
\label{eq:normpsi}
\ee
where \(g : [0,\infty) \rightarrow (0,\infty) \) is a monotonically increasing function which is such that the function $1/g$ is bounded. We admit such a weight factor with hindsight. As we shall see, a suitable choice of $g$ will make a contraction mapping argument possible.

In \eqref{eq:normpsi} we use the notation
\be
	[\psi]^2(x_1^0,x_2^0) ~= \sum_{k=0}^3 \| (\mathcal{D}_k \psi)(x_1^0, \cdot, x_2^0,\cdot)\|^2_{L^2(\R^6,\CC^{16})}
\label{eq:spatialnorm}
\ee
with
\be
	\mathcal{D}_k = \left\{ \begin{array}{cl} 1, &k=0\\ D_1, & k=1\\ D_2, & k=2 \\ D_1 D_2, & k=3 \end{array}\right. 
\label{eq:defdk}
\ee

\paragraph{Remark.} One can see the purpose of integral equation \eqref{eq:inteq} in determining an interacting correction to a solution $\psi^\free$ of the free multi-time Dirac equations $D_i \psi^\free = 0,~i=1,2$. Therefore, it is important to check that sufficiently many solutions of these free equations lie in $\Banach_g$. This is ensured by the following Lemma (see Sec. \ref{sec:freesolutionsinbanach} for a proof).

\begin{lemma}
	Let $\psi^\free$ be a solution of the free multi-time Dirac equations $D_i \psi^\free = 0,~i=1,2$ with initial data $\psi^\free(0,\cdot,0,\cdot) = \psi_0 \in C_c^\infty(\R^6,\CC^{16})$. Furthermore, let $g : [0,\infty) \rightarrow (0,\infty)$ be a monotonically increasing function with $g(t) \rightarrow \infty$ for $t \rightarrow \infty$ and \(g(0)=1\). Then $\psi^\free$ lies in $\Banach_g$.
	\label{thm:freesolutionsinbanach}
\end{lemma}

Given the definition of $A$ on $\mathscr{D}$ as in Sec.\ \ref{sec:aontestfunctions}, we shall now proceed with showing that $A$ is bounded on this space. Furthermore, we show that for a suitable choice of the weight factor $g$ in $\Banach_g$, we can achieve $\| A\| < 1$ on $\mathscr{D}$. This allows to extend $A$ to a contraction on $\Banach_g$ so that the Neumann series $\psi = \sum_{k=0}^\infty A^k \psi^\free$ yields the unique solution of $\psi = \psi^\free + A\psi$.

\section{Results} \label{sec:results}

\subsection{Results for a Minkowski half space} \label{sec:minkhalfspace}

The core of our results is the following Lemma which allows us to control the growth of the spatial norm of $\psi$ with the two time variables.

\begin{lemma} \label{thm:boundsa} 
	Let $\psi \in \mathscr{D}$, $\slashed{\partial}_k = \gamma_k^\mu \partial_{k,\mu},~k=1,2$ and let $K \in C^2(\R^{8},\CC)$ with
	\be
		\| K \| := \sup_{x_1, x_2 \in \frac{1}{2}\M } \max \left\{ |K(x_1,x_2)|, |\slashed{\partial}_1 K(x_1,x_2)|, |\slashed{\partial}_2 K(x_1,x_2)|, |\slashed{\partial}_1 \slashed{\partial}_2 K(x_1,x_2)|\right\} < \infty.
	\label{eq:normk}
	\ee
Then we have:
\begin{align}
		[ A \psi]^2(x_1^0,x_2^0)  ~\leq ~& \|K\|^2 \prod_{j=1,2} \big( \id + 8\mathcal{A}_j(m_j) \big) \, [\psi]^2(x_1^0,x_2^0),
\label{eq:spatialnormapsi}
\end{align}
where $\mathcal{A}_j(m) = \sum_{k=1}^4\mathcal{A}_j^{(k)}(m)$ with $\mathcal{A}_j^{(k)}$ as defined in \eqref{eq:defcurlyoperators}. $[\psi ]^2(x_1^0,x_2^0)$ is understood as a function in $C^\infty\big( (\tfrac{1}{2}\M)^2\big)$ to which the operators in front of it are applied.
\end{lemma}

The proof can be found in Sec.\ \ref{sec:proofboundsa}.

Lemma \ref{thm:boundsa} can now be used to identify (with some trial and error) a suitable weight factor $g$ which allows us to extend $A$ to a contraction on $\Banach_g$. Our main result is:

\begin{theorem}[Existence and uniqueness of dynamics on a Minkowski half space.] \label{thm:minkhalfspace}
	Let $0 < \| K \| < 1$, $\mu = \max \{ m_1,m_2\}$ and
	\begin{align} 
		g(t)& ~=~ \sqrt{1+b t^8}\, \exp(b t^8/16),	\label{eq:defg}\\
    		b& ~=~ \frac{\|K\|^4}{(1-\|K\|)^4} \left(16+\mu^4\right)^4. \label{eq:defb}
\end{align}
	Then for every $\psi^\free \in \Banach_g$, the equation $\psi = \psi^\free + A \psi$ possesses a unique solution $\psi \in \Banach_g$. Here, $A$ is defined in Eqs. \eqref{eq:defa}-\eqref{eq:a4}.
\end{theorem}

The proof is given in Sec.\ \ref{sec:proofminkhalfspace}.

\paragraph{Remarks.}
\begin{enumerate}
	\item Note that Thm. \ref{thm:minkhalfspace} establishes the existence and uniqueness of a global-in-time solution. The non-Markovian nature of the dynamics makes it necessary to prove such a result directly instead of concatenating short-time solutions. The key step in our proof which makes the global-in-time result possible is the suitable choice of the weight factor $g$.
	\item The main condition in Thm.\ \ref{thm:minkhalfspace} is $\| K \| < 1$. This means that the interaction must not be too strong (in a suitable sense). A condition of that kind is to be expected solely because of the contribution $\| (D_1 D_2 (A \psi))(x_1^0,\cdot, x_2^0,\cdot)\|_{L^2} = \| K \psi(x_1^0,\cdot,x_2^0,\cdot)\|_{L^2}$ to $[A \psi](x_1^0,x_2^0)$. Taking our
strategy for setting up the Banach space for granted, we therefore think that one cannot avoid a condition on the interaction strength. Note that conditions on the interaction strength also occur at other places in quantum theory (albeit in a different sense). For example, the Dirac Hamiltonian plus a Coulomb potential is only self-adjoint if the prefactor of the latter is smaller than a certain value.
	\item \textit{Cauchy problem.} Thm.\ \ref{thm:minkhalfspace} shows that $\psi^\free$ uniquely determines the solution $\psi$. However, specifying a whole function in $\Banach_g$ amounts to a lot of data. In case $\psi^\free$ is a solution of the free multi-time Dirac equations $D_1 \psi^\free = 0 = D_2 \psi^\free$ much less data are needed. $\psi^\free$ is then determined uniquely by Cauchy data, and hence $\psi$ is as well. 
Furthermore, if $\psi^\free$ is differentiable, \eqref{eq:inteq} yields
	\be
		\psi(0,\vx_1,0,\vx_2) ~=~ \psi^\free(0,\vx_1,0,\vx_2).
	\label{eq:cauchyproblem}
	\ee
	Thus, Cauchy data for $\psi^\free$ at $x_1^0 = x_2^0 = 0$ are also Cauchy data for $\psi$. The procedure works for arbitrary Cauchy data which are appropriate for the free multi-time Dirac equations.  Note, however, that a Cauchy problem for $\psi$ for times $x_1^0 = t_0 = x_2^0$ with $t_0 > 0$ is not possible. The reason is that $\psi(t_0,\vx_1,t_0,\vx_2) \neq \psi^\free(t_0,\vx_1,t_0,\vx_2)$ in general (and contrary to \eqref{eq:cauchyproblem} the point-wise evaluation may not make sense for $\psi$).
\end{enumerate}

\subsection{Results for a FLRW universe with a Big Bang singularity} \label{sec:flrw}

In this section we show that a Big Bang singularity provides a natural and covariant justification for the cutoff at $t = 0$. As this justification is our main goal, we make the point at the example of a particular class of Friedman-Lema\^itre-Robertson-Walker (FLRW) spacetimes and do not strive to treat more general spacetimes here. The reason for studying these FLRW spacetimes is that they are conformally equivalent to $\frac{1}{2} \M$ \cite{ibison}. Together with the conformal invariance of the massless Dirac operator this allows for an efficient method of calculating the Green's functions which occur in the curved spacetime analog of the integral equation \eqref{eq:inteq}. By doing this, we show that the existence and uniqueness result on these spacetimes can be reduced to Thm.\ \ref{thm:minkhalfspace}.

As shown in \cite{int_eq_curved}, Eq.\ \eqref{eq:inteq} possesses a natural generalization to curved spacetimes $\mathcal{M}$,
\be
	\psi(x_1,x_2) = \psi^\free(x_1,x_2) + \int dV(x_1') \int dV(x_2') \, G_1(x_1,x_1') G_2(x_2,x_2') K(x_1',x_2') \psi(x_1',x_2').
	\label{eq:inteqcurved}
\ee
Here, $dV(x)$ is the spacetime volume element, $S_i$ are (retarded) Green's functions of the respective free wave equation, i.e.\
\be
	D G(x,x') = [-g(x)]^{-1/2} \, \delta^{(4)}(x,x'),
	\label{eq:greensfndefcurved}
\ee
where $g(x)$ is the metric determinant, $D$ the covariant Dirac operator on $\mathcal{M}$, and $\psi$ a section of the tensor spinor bundle over $\mathcal{M} \times \mathcal{M}$.

In order to explicitly formulate \eqref{eq:inteqcurved}, we need to know the detailed form of $S^\ret$. Note that results for general classes of spacetimes showing that $S^\ret$ is a bounded operator on a suitable function space are not sufficient to obtain a strong (global in time) existence and uniqueness result. We therefore focus on the case of a flat FLRW universe where it is easy to determine the Green's functions explicitly. In that case, the metric is given by
\be
	ds^2 = a^2(\eta) [d \eta^2 - d \vx^2]
	\label{eq:metricflrw}
\ee
where $\eta$ is cosmological time and $a(\eta)$ denotes the so-called \textit{scale factor}. The coordinate ranges are given by $\eta \in [0,\infty)$ and $\vx \in \R^3$. For a FLRW universe with a Big Bang singularity, $a(\eta)$ is a continuous, monotonically increasing function of $\eta$ with $a(\eta)=0$, corresponding to the Big Bang singularity. The spacetime volume element reads
\be
	dV(x) = a^4(\eta) \, d \eta\, d^3 \vx.
\ee
The crucial point now is that according to \eqref{eq:metricflrw} the spacetime is globally conformally equivalent to $\frac{1}{2}\M$, with conformal factor
\be
	\Omega(x) = a(\eta).
\ee
In addition, for $m=0$, the Dirac equation is known to be conformally invariant (see e.g.\ \cite{penrose_rindler}). More accurately, consider two spacetimes $\mathcal{M}$ and $\widetilde{\mathcal{M}}$ with metrics
\be
	\widetilde{g}_{ab} = \Omega^2 \, g_{ab}.
\ee
Then the massless Dirac operator $D$ on $\mathcal{M}$ is related to the massless Dirac operator $\widetilde{D}$ on $\widetilde{\mathcal{M}}$ by (see \cite{haantjes}):
\be
	\widetilde{D} = \Omega^{-5/2} \, D \, \Omega^{3/2}.
	\label{eq:diracoperatortrafo}
\ee
This implies the following transformation behavior of the Green's functions:
\be
	\widetilde{G}(x,x') = \Omega^{-3/2}(x) \,\Omega^{-3/2}(x') \, G(x,x').
\ee
One can verify this easily using \eqref{eq:diracoperatortrafo} and the definition of Green's functions on curved spacetimes \eqref{eq:greensfndefcurved}.

Denoting the Green's functions of the Dirac operator on Minkowski spacetime by $G(x,x')=S(x-x')$ and using coordinates $\eta, \vx$ we thus obtain the Green's functions $\widetilde{G}$ on flat FLRW spacetimes as:
\be
	\widetilde{G}(\eta,\vx; \eta', \vx') = a^{-3/2}(\eta) a^{-3/2}(\eta')\, S(\eta-\eta', \vx-\vx').
\ee
With this result, we can write out in detail the multi-time integral equation \eqref{eq:inteqcurved} for massless Dirac particles on flat FLRW spacetimes (using retarded Green's functions):
\begin{align}
	&\psi(\eta_1,\vx_1,\eta_2,\vx_2) = \psi^\free(\eta_1,\vx_1,\eta_2,\vx_2) + a^{-3/2}(\eta_1) a^{-3/2}(\eta_2) \int_0^\infty d \eta_1' \int d^3 \vx_1' \int_0^\infty d \eta_2' \int d^3 \vx_2' \nonumber\\
	&\times a^{5/2}(\eta_1') a^{5/2}(\eta_2') \, S_1^\ret(\eta_1-\eta_1', \vx_1-\vx_1') S_2^\ret(\eta_2-\eta_2',\vx_2-\vx_2') \, (K \psi)(\eta_1',\vx_1',\eta_2',\vx_2').
\end{align}
Note that we can regard $\psi$ as a map $\psi : (\frac{1}{2}\M)^2 \rightarrow \CC^{16}$ as the coordinates $\eta, \vx$ cover the flat FLRW spacetime manifold globally.

It seems reasonable to allow for a singularity of the interaction kernel, i.e.\
\be
	K(\eta_1,\vx_1,\eta_2,\vx_2) = a^{-\alpha}(\eta_1)  a^{-\alpha}(\eta_1) \, \widetilde{K}(\eta_1,\vx_1,\eta_2,\vx_2).
	\label{eq:ksingularity}
\ee
Here, $\alpha \geq 0$. The singular behavior is motivated by that of the Green's functions of the conformal wave equation\footnote{The conformal wave equation reads $(\Box- R/6)\phi = 0$ where $\Box$ is the d'Alembertian and $R$ the Ricci scalar of the respective spacetime.}. Recall from the introduction that the most natural interaction kernel on $\frac{1}{2}\M$ would be $K(x_1,x_2)\propto \delta((x_1-x_2)_\mu(x_1-x_2)^\mu)$ which is a Green's function of the wave equation -- a concept that can be generalized to curved spacetimes using the conformal wave equation. Now, under conformal transformations, Green's functions of that equation transform as \cite{john}
\be
	\widetilde{G}(x,x') = \Omega^{-1}(x) \,\Omega^{-1}(x') \, G(x,x'),
\ee
which corresponds to $\alpha = 1$ in \eqref{eq:ksingularity}.

Considering \eqref{eq:ksingularity}, our integral equation becomes:
\begin{align}
	&\psi(\eta_1,\vx_1,\eta_2,\vx_2) = \psi^\free(\eta_1,\vx_1,\eta_2,\vx_2) + a^{-3/2}(\eta_1) a^{-3/2}(\eta_2) \int_0^\infty d \eta_1' \int d^3 \vx_1' \int_0^\infty d \eta_2' \int d^3 \vx_2'\nonumber\\
	\times &a^{5/2-\alpha}(\eta_1') a^{5/2-\alpha}(\eta_2') \, S_1^\ret(\eta_1-\eta_1', \vx_1-\vx_1') S_2^\ret(\eta_2-\eta_2',\vx_2-\vx_2') \, (\widetilde{K} \psi)(\eta_1',\vx_1',\eta_2',\vx_2').
\label{eq:inteqcurvedexplicit}
\end{align}
Apart from the scale factors which produce a certain singularity of $\psi$ for $\eta_1,\eta_2 \rightarrow 0$, this integral equation has the form of \eqref{eq:inteq} on $\frac{1}{2}\M$. Indeed, we can use the transformation
\be
	\chi(\eta_1,\vx_1,\eta_2,\vx_2) = a^{3/2}(\eta_1) a^{3/2}(\eta_2)\, \psi(\eta_1,\vx_1,\eta_2,\vx_2)
	\label{eq:psichi}
\ee
to transform the two equations into each other. We arrive at the following result.

\begin{theorem}[Existence and uniqueness of dynamics on a flat FLRW universe] \label{thm:flrw}
	Let, $0 \leq \alpha \leq 1$ and let $a : [0,\infty) \rightarrow [0,\infty)$ be a differentiable function with $a(0)=0$ and $a(\eta) >0$ for $\eta>0$. Moreover, assume that $\widetilde{K} \in C^2 \left( ([0,\infty)\times \R^3)^2,\CC\right)$ with
	\be
		\| a^{1-\alpha}(\eta_1) a^{1-\alpha}(\eta_2) \, \widetilde{K} \| <1.
	\label{eq:ktildecondition}
	\ee
 Then for every $\psi^\free$ with $a^{3/2}(\eta_1) a^{3/2}(\eta_2) \psi^\free \in \Banach_g$, \eqref{eq:inteqcurvedexplicit} has a unique solution $\psi$ with $a^{3/2}(\eta_1) a^{3/2}(\eta_2)\psi \in \Banach_g$ (and with \(g\) as in Thm. \ref{thm:minkhalfspace}).
\end{theorem}

\begin{proof}
	Multiplying \eqref{eq:inteqcurvedexplicit} with $a^{3/2}(\eta_1) a^{3/2}(\eta_2)$ and using the relation yields
	\begin{align}
	&\chi(\eta_1,\vx_1,\eta_2,\vx_2) = \chi^\free(\eta_1,\vx_1,\eta_2,\vx_2) +\int_0^\infty d \eta_1' \int d^3 \vx_1 \int_0^\infty d \eta_2'~a^{1-\alpha}(\eta_1') a^{1-\alpha}(\eta_2') \nonumber\\
	&\times  \, S_1^\ret(\eta_1-\eta_1', \vx_1-\vx_1') S_2^\ret(\eta_2-\eta_2',\vx_2-\vx_2') \, (\widetilde{K} \chi)(\eta_1',\vx_1',\eta_2',\vx_2').
\label{eq:inteqcurvedexplicit2}
\end{align}
This equation has the form of \eqref{eq:inteq} on $\frac{1}{2}\M$ with $K$ replaced by $a^{1-\alpha}(\eta_1') a^{1-\alpha}(\eta_2') \widetilde{K}$. Thus, using the same distributional understanding of the Green's functions as before, Thm.\ \ref{thm:minkhalfspace} yields the claim. \qed 
\end{proof}

\paragraph{Remarks.}
\begin{enumerate}
	\item Both $\psi^\free$ and $\psi$ have a singularity proportional to $a^{-3/2}(\eta_1)a^{-3/2}(\eta_2)$ for $\eta_1, \eta_2 \rightarrow 0$.
	\item For $\alpha < 1$, $\widetilde{K}$ has to compensate the singularities caused by $a^{-3/2}(\eta_1) a^{-3/2}(\eta_2)$ in order for \eqref{eq:ktildecondition} to hold. In the most natural case $\alpha = 1$, however, $\widetilde{K}$ only needs to satisfy $\| \widetilde{K} \| < 1$, i.e.\, the same condition as for $K$ in Thm.\ \ref{thm:minkhalfspace}. 
	\item Let $\chi^\free = a^{3/2}(\eta_1) a^{3/2}(\eta_2) \psi^\free$ be differentiable and let $\chi$ be the unique solution of \eqref{eq:inteqcurvedexplicit2}. Then, by \eqref{eq:inteqcurvedexplicit2}, we have:
	\be
		\chi^\free(0,\vx_1,0,\vx_2) = \chi(0,\vx_1,0,\vx_2),
	\ee
	i.e.\, $\chi$ satisfies a Cauchy problem "at the Big Bang".
	\item Remarkably, Thm.\ \ref{thm:flrw} covers a class of manifestly covariant, interacting integral equations in 1+3 dimensions. Then the interaction kernel $\widetilde{K}$ has to be covariant as well. A class of examples (see also \cite{int_eq_curved}) is given by $\alpha = 1$ and 
\be
	\widetilde{K}(x_1,x_2) = \left\{ \begin{array}{cl} f(d(x_1,x_2))& \text{if } x_1, x_2 \text{ are time-like related}\\ 0 &  \text{else}, \end{array} \right.
\ee
where $d(x_1,x_2) = (|\eta_1-\eta_2| -|\vx_1-\vx_2|) \int_0^1 d \tau \, a(\tau \eta_1 + (1-\tau) \eta_2)$ denotes the time-like distance of the spacetime points $x_1 = (\eta_1,\vx_1)$ and $x_2 = (\eta_2,\vx_2)$, and $f$ is an arbitrary smooth function which leads to $\| \widetilde{K} \| < 1$.

\end{enumerate}

\section{Proofs} \label{sec:proofs}

\subsection{Proof of Lemma \ref{thm:freesolutionsinbanach}}

\label{sec:freesolutionsinbanach}

Consider a solution $\psi$ of $D_i \psi^\free = 0,~i=1,2$ for compactly supported initial data at $x_1^0 = 0 = x_2^0$. As the Dirac equation has finite propagation speed, $\psi^\free$  is spatially compactly supported for all times. Without loss of generality we may assume \(\|\psi^\free (t_1,\cdot,t_2,\cdot)\|_{L^2(\mathbb{R}^6)}=1\) for all times \(t_1,t_2\), so it follows that also \([\psi^\free](t_1,t_2)=1\). In the following we will construct a sequence of test functions \((\psi_m)_{m\in\mathbb{N}}\) satisfying \(\psi_m \xrightarrow{\|\cdot\|_g ,~m\rightarrow \infty} \psi^\free\). Let \(\eta : \mathbb{R}\rightarrow \mathbb{R}\) be zero for arguments less than \(0\) and greater than \(1\) and in between given by (see also Fig. \ref{fig:eta})
\be
	\eta(t)=  \exp\left(-\frac{1}{t} \exp\left({\frac{1}{t-1}}\right)\right).
\ee
\begin{figure}
	\centering
	\begin{tikzpicture}
  \begin{axis}[
      xlabel={$t$}, 
      ylabel={$\eta$},
      samples=500,
      xmin=-0,
      xmax=1,
      ymin=-0.1,
      ymax=1.1,
      enlarge y limits=false,
      axis x line=middle,
      axis y line=middle,
      axis line style={-},
      width=10. cm,
      height=6. cm,
      yticklabel style={xshift=7.5mm},
      ylabel style={xshift=-3mm,yshift=5mm},
      xlabel style={xshift=5mm,yshift=-2.4mm},
      ticklabel style={font=\small}
    ]
    \addplot[black,thick,domain=-0.1:0.001](x,0);
    \addplot [black,thick,domain=0.001:0.999] (x, {pow(e,-1/x * pow(e,1/(x-1)))});
    \addplot[black,thick,domain=0.999:1.1] (x,1);
    \end{axis}
\end{tikzpicture}
	\caption{The function \( \eta(t)\).}
	\label{fig:eta}
\end{figure}
 Note that $\eta$ is smooth and monotonically increasing. Next, we define for every \(m\in\mathbb{N}\)
 \begin{equation}
 \psi^\free_m(t_1,\vx_1, t_1, \vx_2) ~:=~ e^{-\eta(t_1-m) (t_1-m)} e^{-\eta(t_2-m)(t_2-m)} \psi^\free(t_1,\vx_1, t_2, \vx_2).
 \end{equation}
 This function is smooth and decreases rapidly in all variables and thus lies in \(\mathscr{D}\). Now we estimate
  \(\|\psi^\free-\psi_m\|_g\). Pick \(n\in\mathbb{N}\). First consider \(\|\psi^\free -\psi_n\|_{L^2(\mathbb{R}^6)}(t_1,t_2)\). This function is identically zero for all \(t_1<n\) and \(t_2<n\), so we obtain the estimate
 \begin{align}
    & \sup_{t_1,t_2>0} \frac{1}{g(t_1)^2g(t_2)^2} \|\psi^\free- \psi_n\|^2_{L^2(\mathbb{R}^6)}\nonumber\\
     =&\sup_{t_1,t_2>0} \frac{1}{g(t_1)^2g(t_2)^2} \left|1-e^{-\eta(t_1-n)(t_1-n)}e^{-\eta(t_2-n)(t_2-n)}\right| ~\le~\frac{1}{g(n)^2}.
\end{align}
For the other terms we use that \(\psi^\free\) solves the free Dirac equation in each variable and that \(\sup_{t>0}\partial_{t} e^{-\eta(t)t}=:\alpha<\infty\) is realized for some positive value of \(t\). So we find for \(i\in \{0,1\}\):
 \begin{align}
& \sup_{t_1,t_2>0}\frac{1}{g(t_1)^2g(t_2)^2}\|D_i (\psi^\free-\psi_n)\|^2_{L^2(\mathbb{R}^6)}(t_1,t_2)\nonumber\\
= &\sup_{t_1,t_2>0}\frac{1}{g(t_1)^2g(t_2)^2}
 \|\gamma_{i}^0 \psi^\free(t_1,\cdot,t_2,\cdot) e^{-\eta(t_{3-i}-n)(t_{3-i}-n)}\partial_{t_i}e^{-\eta(t_i-n)(t_i-n)} \|^2_{L^2(\mathbb{R}^6)}~ \le~ \frac{\alpha}{g(n)^2}.
 \end{align}
For the first inequality it has been used that the factor with a derivative vanishes for \(t_i<n\). 

An analogous estimate repeated for the \(D_1D_2\)-term yields
\begin{equation*}
     \sup_{t_1,t_2>0} \frac{1}{g(t_1)^2g(t_2)^2}\|D_1D_2(\psi^\free - \psi_n)\|^2_{L^2(\mathbb{R}^6)}(t_1,t_2)~\le~ \frac{\alpha^2}{g(n)^4} ~\le~ \frac{\alpha^2}{g(n)^2}.
 \end{equation*}
 All in all, adding the estimates and taking the square root we find \(\|\psi^\free - \psi_n\|_g \le \frac{1+\alpha}{g(n)}\), which together with the asymptotic behavior of \(g\) implies convergence. It follows that the free solution \(\psi^\free\) can be approximated by Cauchy sequences in \(\mathscr{D}\) and hence is contained in \(\mathscr{B}_g\) which, we recall, has been defined as the completion of $\mathscr{D}$ with respect to $\| \cdot \|_g$. \qed

\subsection{Proof of Lemma \ref{thm:boundsa}} \label{sec:proofboundsa}

Throughout the following subsections, let $\psi \in \mathscr{D}$ and $K : \R^8 \rightarrow \CC$ be a smooth function. Furthermore define \(\delta:= 1-\|K\|^2>0, \mu=\text{max}\{m_1,m_2\}\) and let \(g\) be as in the statement of Thm. \ref{thm:minkhalfspace}.

We begin with some lemmas which are useful for estimating $[A\psi]^2(x_1^0,x_2^0)$.

\begin{lemma} \label{thm:estimatetensoroperators}
	Let the following operators be defined on $C([0,\infty))$:
	\begin{align}
		\big(\mathcal{A}^{(1)}(m) f\big)(t) ~&=~ t \int_0^{t} d \rho ~ (t-\rho)^2\,  f(\rho),\nonumber\\
		\big(\mathcal{A}^{(2)}(m) f \big)(t) ~&=~ \frac{m^4 t^4}{2^4 \,3^2} \int_0^{t} d\rho ~(t-\rho)^3 \, f(\rho),\nonumber\\
		\big(\mathcal{A}^{(3)}(m) f\big)(t) ~&=~ t^2 \, f(0),\nonumber\\
		\big(\mathcal{A}^{(4)}(m) f \big)(t) ~&=~ \frac{m^4 t^6}{2^2\, 3^2} \, f(0).
	\label{eq:defcurlyoperators}
	\end{align}
	Then, for $j=1,2$ and $k=1,2,3,4$, we define the operator $\mathcal{A}_j^{(k)}(m)$ acting on functions $\phi \in C([0,\infty)^2)$ by letting $\mathcal{A}^{(k)}(m)$ act on the $j$-th variable of $\phi(t_1,t_2)$.
Then we have for all $\psi \in \mathscr{D}$, all $m_1,m_2\geq 0$ and all $k,l=1,2,3,4$:
\be
		\left\| A_1^{(k)}(m_1) A_2^{(l)}(m_2)  \psi(t_1,\cdot,t_2,\cdot) \right\|^2_{L^2} ~\leq~ \mathcal{A}_j^{(k)}(m_1) \mathcal{A}_j^{(l)}(m_2)\, \|\psi(t_1,\cdot,t_2,\cdot)\|^2_{L^2}.\label{eq:aestimate}
	\ee
Here, it is understood that the operators $\mathcal{A}_j^{(k)}$ are applied to the functions defined by the norms which follow them, e.g.\, $\mathcal{A}_1^{(4)}(m_1)\, \|\psi(t_1,\cdot,t_2,\cdot)\|^2_{L^2} = \frac{m^4_1 t_1^6}{2^2\, 3^2} \, \|\psi(0,\cdot,t_2,\cdot)\|^2_{L^2}$.
\end{lemma}

\begin{proof}
	We prove \eqref{eq:aestimate} for $k=1$, $ l=2$ and $k=3$, $l=4$. The remaining cases can be treated in the same way. We begin with $k=1$, $l=2$, using $|J_1(x)/x| \leq \frac{1}{2}$:
\begin{align}
	&\| A_1^{(1)}(m_1) A_2^{(2)}(m_2)\, \psi(x_1^0,\cdot,x_2^0,\cdot) \|^2_{L^2} ~= \frac{m_2^2}{(4\pi)^4} \int_{\R^3 \times \R^3} \!\!\! d^3 \vx_1 \, d^3 \vx_2 \nonumber\\
&~~~~~\times \left| \int_{B_{x_1^0}(0)} \!\!\! d^3\vy_1 \int_{-x_2^0}^0 dy_2^0 \int_{B_{|y_2^0|}(0)} \!\!\! d^3 \vy_2 ~\frac{1}{|\vy_1|}  \frac{J_1(m_2\sqrt{y_2^2})}{\sqrt{y_2^2}}\psi(x_1+y_1,x_2+y_2)|_{y^0_1=-|\vy_1|} \right|^2\nonumber\\
&\leq~ \frac{m_2^2}{(4\pi)^4} \int_{\R^3 \times \R^3}  \!\!\! d^3 \vx_1 \, d^3 \vx_2 \left( \int_{B_{x_1^0}(0)} \!\!\! d^3\vy_1 \int_{-x_2^0}^0 dy_2^0 \int_{B_{|y_2^0|}(0)} \!\!\! d^3 \vy_2~ \frac{1}{|\vy_1|^2} \left| \frac{J_1(m_2\sqrt{y_2^2})}{\sqrt{y_2^2}} \right|^2 \right)\nonumber\\
&~~~~~\times \left( \int_{B_{x_1^0}(0)} \!\!\! d^3\vy_1 \int_{-x_2^0}^0 dy_2^0 \int_{B_{|y_2^0|}(0)} \!\!\! d^3 \vy_2 ~|\psi|^2(x_1+y_1,x_2+y_2)|_{y^0_1=-|\vy_1|} \right)\nonumber\\
&\leq~  \frac{m_2^2}{(4\pi)^4} \int_{\R^3 \times \R^3}  \!\!\! d^3 \vx_1 \, d^3 \vx_2 ~4\pi x_1^0 \left( \frac{\pi m_2^2 (x_2^0)^4}{12}\right) \nonumber\\
&~~~~~\times \left( \int_{B_{x_1^0}(0)} \!\!\! d^3\vy_1 \int_{-x_2^0}^0 dy_2^0 \int_{B_{|y_2^0|}(0)} \!\!\! d^3 \vy_2 ~|\psi|^2(x_1+y_1,x_2+y_2)|_{y^0_1=-|\vy_1|} \right)\nonumber\\
&\leq~  \frac{m_2^4 \, x_1^0 \,  (x_2^0)^4}{3 \pi^2\, 2^8 } \int_{\R^3 \times \R^3}  \!\!\! d^3 \vx_1 \, d^3 \vx_2  \int_{B_{x_1^0}(0)} \!\!\! d^3\vy_1 \int_{-x_2^0}^0 dy_2^0 \int_{B_{|y_2^0|}(0)} \!\!\! d^3 \vy_2\nonumber\\
&~~~~~\times |\psi|^2(x_1^0 - |\vy_1|,\vx_1+\vy_1,x_2^0 + y_2^0,\vx_2+\vy_2).
\label{eq:a1a2estimate1}
\end{align}
Exchanging the $x$ and $y$ integrals yields:
\begin{align}
&\eqref{eq:a1a2estimate1} \leq  \frac{m_2^4 \, x_1^0\,  (x_2^0)^4}{3 \pi^2\, 2^8 } \int_{B_{x_1^0}(0)} \!\!\! d^3\vy_1 \int_{-x_2^0}^0 dy_2^0 \int_{B_{|y_2^0|}(0)} \!\!\! d^3 \vy_2 ~\| \psi(x_1^0 - |\vy_1|,\cdot,x_2^0 + y_2^0,\cdot)\|_{L^2} \nonumber\\
&\leq~  \frac{m_2^4 \, x_1^0\,  (x_2^0)^4}{3 \pi^2\, 2^8 } \, 4\pi \int_{0}^{x_1^0} \!\!\! d r_1 ~r_1^2 \int_{-x_2^0}^0 dy_2^0 ~\frac{4\pi}{3} |y_2^0|^3 \, \| \psi(x_1^0 - |\vy_1|,\cdot,x_2^0 + y_2^0,\cdot)\|_{L^2} \nonumber\\
&\leq~  \frac{m_2^4 \, x_1^0 \,  (x_2^0)^4}{2^4 \, 3^2} \int_{0}^{x_1^0} \!\!\! d \rho_1 ~(x_1^0-\rho_1)^2 \int_{0}^{x_2^0} d\rho_2 ~(x_2^0-\rho_2)^3 \, \| \psi(\rho_1,\cdot,\rho_2,\cdot)\|_{L^2} \nonumber\\
&=~ \mathcal{A}_1^{(1)}(m_1) \mathcal{A}_2^{(2)}(m_2)\, \|\psi(x_1^0,\cdot,x_2^0,\cdot)\|^2_{L^2}.
\end{align}

Next, we turn to the case $k=3, l=4$. Using that the modulus of the largest eigenvalue of $\gamma^0$ is 1, we obtain:
\begin{align}
	&\| A_1^{(3)}(m_1) A_2^{(4)}(m_2)\, \psi(x^0_1,\cdot,x_2^0,\cdot) \|^2_{L^2} ~\leq~ \frac{m_2^2}{(4\pi)^4 (x_1^0)^2} \int_{\R^3\times \R^3} \!\!\! d^3 \vx_1 \, d^3 \vx_1\nonumber\\
 &~~~~~\times \left|\int_{\partial B_{x_1^0}(0)} \!\!\! d \sigma(\vy_1) \int_{B_{x_2^0}(0)} \!\!\! d^3 \vy_2 ~\frac{J_1\left( m_2\sqrt{(x_2^0)^2-\vy_2^2}\right)}{\sqrt{(x_2^0)^2-\vy_2^2}} |\psi|(0,\vx_1+\vy_2,0,\vx_2+\vy_2) \right|^2\nonumber\\
&\leq~ \frac{m_2^4}{(4\pi)^4 (x_1^0)^2} \int_{\R^3\times \R^3} \!\!\! d^3 \vx_1 \, d^3 \vx_1 \left(\int_{\partial B_{x_1^0}(0)} \!\!\! d \sigma(\vy_1) \int_{B_{x_2^0}(0)} \!\!\! d^3 \vy_2 \left|\frac{J_1\left( m_2\sqrt{(x_2^0)^2-\vy_2^2}\right)}{m_2\sqrt{(x_2^0)^2-\vy_2^2}}\right|^2 \right)\nonumber\\
&~~~~~ \times \left( \int_{\partial B_{x_1^0}(0)} \!\!\! d \sigma(\vy_1) \int_{\partial B_{x_2^0}(0)} \!\!\! d \sigma(\vy_2) ~|\psi|^2(0,\vx_1+\vy_2,0,\vx_2+\vy_2)\right)\nonumber\\
&=~ \frac{m_2^4}{(4\pi)^4 (x_1^0)^2} \, 4\pi (x_1^0)^2 \, \frac{\pi (x_2^0)^3}{3} \int_{\R^3\times \R^3} \!\!\! d^3 \vx_1 \, d^3 \vx_1 \int_{\partial B_{x_1^0}(0)} \!\!\! d \sigma(\vy_1) \int_{B_{x_2^0}(0)} \!\!\! d^3 \vy_2 \nonumber\\
&~~~~~\times ~|\psi|^2(0,\vx_1+\vy_2,0,\vx_2+\vy_2).
\label{eq:a3a4estimate1}
\end{align}
Exchanging the order of the $x$ and $y$ integrals yields:
\begin{align}
	\eqref{eq:a3a4estimate1} ~&=~ \frac{m_2^4}{(4\pi)^4} \pi (x_2^0)^2\int_{\partial B_{x_1^0}(0)} \!\!\! d \sigma(\vy_1) \int_{B_{x_2^0}(0)} \!\!\! d^3 \vy_2 ~\|\psi(0,\cdot,0,\cdot)\|^2_{L^2}\nonumber\\
&=~ \frac{m_2^4 \, (x_1^0)^2 \, (x_2^0)^6}{2^2 \, 3^2} \, \|\psi(0,\cdot,0,\cdot)\|^2_{L^2}\nonumber\\
&= ~ \mathcal{A}_1^{(3)}(m_1) \mathcal{A}_2^{(4)}(m_2) \,  \|\psi(x_1^0,\cdot,x_2^0,\cdot)\|^2_{L^2}.
\end{align} \qed
\end{proof}

\begin{lemma} \label{thm:spatialnormestimates}
	For $j=1,2$ let $\mathcal{A}_j(m) = \sum_{k=1}^4 \mathcal{A}_j^{(k)}(m)$. Then the following estimates hold:
\begin{align}
		& \!\| (A \psi)(x_1^0,\cdot,x_2^0,\cdot)\|^2_{L^2} ~\leq ~64 \, \|K\|^2 \mathcal{A}_1(m_1) \mathcal{A}_2(m_2) \, [\psi]^2(x_1^0,x_2^0),
	\label{eq:apsiestimate}\\
		&\! \| (D_1(A \psi))(x_1^0,\cdot,x_2^0,\cdot) \|^2_{L^2}~\leq~8 \, \| K \|^2 \, \mathcal{A}_2(m_2) \, [\psi](x_1^0,x_2^0),
	\label{eq:d1apsiestimate}\\
		&\! \| (D_2(A \psi))(x_1^0,\cdot,x_2^0,\cdot) \|^2_{L^2} ~\leq~8 \, \| K \|^2 \, \mathcal{A}_1(m_1) \, [\psi](x_1^0,x_2^0),
	\label{eq:d2apsiestimate}\\
		&\! \| (D_1 D_2(A \psi))(x_1^0,\cdot,x_2^0,\cdot) \|^2_{L^2} ~\leq~ \| K \|^2 \, [\psi]^2(x_1^0,x_2^0),
	\label{eq:d1d2apsiestimate}
	\end{align}
	where $[\psi]^2(x_1^0,x_2^0)$ is regarded as a function of $x_1^0,x_2^0$ to which the operators in front of it are applied.
\end{lemma}

\begin{proof}
	We start with \eqref{eq:apsiestimate}. Recalling \eqref{eq:defa}, the expression $A\psi$ contains terms such as $D_1 D_2 (K \psi)$ and $D_i(K\psi)$, $i=1,2$. Recalling also the definition of $\mathcal{D}_k$ (Eq.\ \eqref{eq:defdk}), we have:
\be
	D_1 D_2 (K\psi) = \sum_{k=0}^3 (\nabla_{3-k}K)(\mathcal{D}_k \psi)
\ee
with
\be
\nabla_k := \left\{\begin{array}{cl} 1, & k=0\\ i\slashed{\partial}_1, & k=1 \\ i \slashed{\partial}_2, & k=2 \\ -\slashed{\partial}_1\slashed{\partial}_2, & k=3. \end{array}\right.
\ee
Hence, noting \eqref{eq:normk}:
\be
	|D_1D_2 \psi| \leq \| K \| \sum_{k=0}^3 |\mathcal{D}_k \psi|.
\ee
Similarly, we find:
\be
	D_i (K\psi) \leq \| K\| \sum_{k=0}^3 |\mathcal{D}_k \psi|,~~i=1,2.
\ee
Considering the definition of $A_j^{(k)}(m),$ $j=1,2$, $k=1,2,3,4$ it follows that
\begin{align}
	|A\psi| ~\leq ~\|K\| \sum_{k=0}^3  & \prod_{j=1,2} \left( A_j(m_j)^{(1)} + A_j^{(2)}(m_j) + A_j^{(3)}(m_j) + A_j^{(4)}(m_j)\right)  |\mathcal{D}_k \psi|.
\end{align}
In slight abuse of notation, we here use the same symbols for the operators $A_j^{(k)}(m)$ acting on functions with and without spin components.

The idea now is to make use of lemma lemma \ref{thm:estimatetensoroperators}. In order to be able to apply the lemma, we first note that by Young's inequality for $a_1,...,a_N \in \R$, we have $\left(\sum_{i=1}^N a_i\right)^2 \leq N \sum_{i=1}^N a_i^2$ and thus:
\be
	|A\psi (x_1,x_2)|^2 ~\leq~ 64\, \|K\|^2 \sum_{i,j=1}^4 \sum_{k=0}^3 \big| A_1^{(i)}(m_1) A_2^{(j)}(m_2) \, |\mathcal{D}_k \psi|\big|^2.
\ee
Integrating over this expression and using lemma \ref{thm:estimatetensoroperators}, we obtain:
\be
		\| (A \psi)(x_1^0,\cdot,x_2^0,\cdot)\|^2_{L^2} ~\leq ~64 \, \|K\|^2\sum_{i,j=1}^4 \sum_{k=0}^3\mathcal{A}_1^{(i)}(m_1) \mathcal{A}_2^{(j)}(m_2) \, \| (\mathcal{D}_k \psi)(x_1^0,\cdot,x_2^0,\cdot) \|^2_{L^2}.
\ee	
Recalling the definition of $[\psi]^2(x_1^0,x_2^0)$, Eq.\ \eqref{eq:spatialnorm} yields \eqref{eq:apsiestimate}.

Next, we turn to \eqref{eq:d1apsiestimate}. We start from the initial form of the integral equation \eqref{eq:inteq} and use that as a distributional identity on test functions $\psi \in \mathscr{D}_T$, we have $D_1 S^\ret(x_1-x_1') = \delta^{(4)}(x_1-x_1')$. Thus, we obtain:
\be
	(D_1 A \psi)(x_1,x_2) = \int_{\tfrac{1}{2} \M} d^4 x_2' ~S_2^\ret(x_2-x_2') (K \psi)(x_1,x_2').
\ee
Proceeding similarly as for \eqref{eq:defa} we rewrite this as:
\be
	D_1 (A \psi) ~=~ \left( A_2^{(1)}(m_2)\, D_2 + A_2^{(2)}(m_2) D_2 + A_2^{(3)}(m_2) + A_2^{(4)}(m_2)\right) (K\psi).
\ee
Considering the form of $A_j^{(k)}(m_j)$ this implies:
\be
	|D_1 (A \psi)| ~\leq~  \| K \| \sum_{i=1}^4 \sum_{k\in \{ 0,2\}} A_2^{(i)}(m_2)\, |\mathcal{D}_k \psi|.
\ee
We now square and use Young's inequality, finding:
\be
	|D_1 (A \psi)|^2 ~\leq~  8 \, \| K \|^2 \sum_{i=1}^4 \sum_{k\in \{ 0,2\}} A_2^{(i)}(m_2) \, | \mathcal{D}_k \psi|^2.
\ee
Integrating and using lemma \ref{thm:estimatetensoroperators} yields:
\begin{align}
	\| D_1 (A \psi)(x_1^0,\cdot,x_2^0,\cdot)\|^2_{L^2} ~\leq~ 8 \, \| K \|^2 \sum_{i=1}^4 \sum_{k\in \{ 0,2\}} \mathcal{A}_2^{(i)}(m_2)\, \| (\mathcal{D}_k \psi)(x_1^0,\cdot,x_2^0,\cdot) \|^2_{L^2}.
\end{align}
Adding the terms with $k=1,3$ and using the definition of $[\psi]^2(x_1^0,x_2^0)$ gives us \eqref{eq:d1apsiestimate}.

The estimate \eqref{eq:d2apsiestimate} follows in an analogous way.

Finally, for \eqref{eq:d1d2apsiestimate} we also start from the initial integral equation \eqref{eq:inteq} and use $D_i S_i^\ret(x_i-x_i') = \delta^{(4)}(x_i-x_i')$. This results in:
\be
	D_1 D_2 (A \psi) = K \psi.
\ee
Squaring and integrating gives us:
\be
	\| D_1 D_2 (A \psi)(x_1^0,\cdot,x_2^0,\cdot) \|^2 ~ \leq ~\| K \|^2 \, \| \psi(x_1^0,\cdot,x_2^0,\cdot) \|^2_{L^2} ~\leq~ \| K \|^2 \, [\psi]^2(x_1^0,x_2^0),
\ee
which yields \eqref{eq:d1d2apsiestimate}. \qed
\end{proof}

These estimates are the core of:

\begin{proof}[Proof of Lemma\ \ref{thm:boundsa}:]
	
	We use lemma \ref{thm:spatialnormestimates} together with the definition of $[\psi]^2(x_1^0,x_2^0)$ to obtain:
	\be
		[A \psi]^2(x_1^0,x_2^0) ~\leq~ \eqref{eq:apsiestimate} + \eqref{eq:d1apsiestimate}+ \eqref{eq:d2apsiestimate}+ \eqref{eq:d1d2apsiestimate}.
	\ee
	Summarizing the operators into a product yields \eqref{eq:spatialnormapsi}.
\qed
\end{proof}

\subsection{Proof of Theorem \ref{thm:minkhalfspace}} \label{sec:proofminkhalfspace}

In order to prove Thm. \ref{thm:minkhalfspace}, we combine the previous estimates to show that \(\|A\|<1\), first on test functions $\psi \in \mathscr{D}$ and by linear extension also on the whole of $\Banach_g$ . We start with Eq. \eqref{eq:spatialnormapsi} of Thm. \ref{sec:minkhalfspace} using the definition of \(\mathcal{A}_j\) for \(j=1,2\), as well as the following estimate, valid for all \(\psi \in \mathscr{D}, t_1,t_2>0\):
\be 
	[\psi](t_1,t_2) ~=~ [\psi](t_1,t_2) \, \frac{g(t_1)g(t_2)}{g(t_1)g(t_2)} ~\le~ \|\psi\|_g\, g(t_1)g(t_2).
\ee 
Using this in \eqref{eq:spatialnormapsi} yields:
\begin{align}
\|A\psi\|^2_g ~&\le~ \sup_{x_1^0,x_2^0>0}\frac{1}{(g(x_1^0)g(x_2^0))^2} \,  \|K\|^2 \prod_{j=1,2} \left( \id + 8\mathcal{A}_j(m_j) \right) [\psi]^2(x_1^0,x_2^0),\\
&\le~ \sup_{x_1^0,x_2^0>0} \frac{\|\psi\|^2}{(g(x_1^0)g(x_2^0))^2} \,  \|K\|^2 \prod_{j=1,2} \left( \id + 8\mathcal{A}_j(m_j) \right) \, (g^2\otimes g^2)(x_1^0,x_2^0),\\
&\le~ \|K\|^2\, \|\psi\|^2_g \left( \sup_{t>0}\frac{1}{g(t)^2} \left( \id + 8\mathcal{A}(\mu) \right)g^2 (t) \right)^2,
\label{est:NormA1}
\end{align}
where $\mu = \max \{ m_1,m_2\}$ and $\mathcal{A}(\mu) = \sum_{k=1}^4 \mathcal{A}^{(k)}(\mu)$ with $ \mathcal{A}^{(k)}(\mu)$ as in \eqref{eq:defcurlyoperators}.

Next, we shall estimate the term in the big round bracket. To this end, we first note some special properties of \(g^2\), which motivated choosing $g$ as in \eqref{eq:defg}.
\begin{lemma}\label{lem:int}
 For all \(t>0\), we have
\be 
	\int_0^t d\tau \, g^2(\tau) = \frac{t}{1+b t^8} \, g^2(t).
\ee 
\end{lemma}
\begin{proof}[Proof:] Differentiating the right side of the equation and using the concrete function \(g^2\) as in \eqref{eq:defg} shows that it is, indeed, the anti-derivative of \(g^2\). Since this function vanishes at \(t=0\), the claim follows. \qed
\end{proof}

\begin{lemma}\label{lem:someidentities}
For $c<8$ we have
\be
	\sup_{t>0} \frac{t^c}{1+b t^8}  ~=~\frac{c}{8}\,  b^{-c/8} \left( \frac{8}{c}-1\right)^{1-c/8},
	\label{eq:firstequality}
\ee
and furthermore for $c=8$:
\be
	\sup_{t>0} \frac{t^8}{1+b t^8} ~=~\frac{1}{b}.
	\label{eq:secondequality}
\ee
\end{lemma}

\begin{proof}
To prove \eqref{eq:firstequality}, considering the shape of the function $h(t)= t^c/(1+b t^8)$ we find that the supremum is in fact a maximum which is located at \(t = b^{-1/8} \left(8/c-1\right)^{-1/8}\). Inserting this back into the function $h(t)$ yields \eqref{eq:firstequality}.  \eqref{eq:secondequality} follows from $\frac{t^8}{1+b t^8} = \frac{1}{b} \frac{1}{1/(b t^8) +1} \leq \frac{1}{b}$. \qed
\end{proof}

\paragraph{Proof of Thm. \ref{thm:minkhalfspace}:} Applying Lemma \ref{lem:int} to \(\mathcal{A}(\mu) \,g^2\) yields:
\begin{align}
    \big(\mathcal{A}^{(1)}(\mu) \, g^2\big)(t)&~=~ t \int_0^t d\rho \, (t-\rho)^2\, g^2(\rho)~\le~ t^3 \int_0^t d\rho \,g^2(\rho) = \frac{t^4}{1+b t^8} \,g^2(t),\nonumber\\
     \big(\mathcal{A}^{(2)}(\mu) \, g^2\big)(t)&~=~ \frac{\mu^4\, t^4}{2^4\, 3^2} \int_0^t d\rho\, (t-\rho)^3\, g^2(\rho)~\le~ \frac{\mu^4\, t^8}{2^4 \,3^2} \frac{g^2(t)}{1+b t^8},\nonumber\\
    \big( \mathcal{A}^{(3)}(\mu) \, g^2\big) (t) &~=~ t^2,\nonumber\\
    \big( \mathcal{A}^{(4)}(\mu) \, g^2 \big)(t) &~=~ \frac{\mu^4\, t^6}{2^2\, 3^2}.
\end{align}
Multiplying with \(1/g^2(t)\) and using Lemma \ref{lem:someidentities} as well as $1/g(t)^2 \le (1+b t^8)^{-1}$, we find:
\begin{align}
    g^{-2}(t)\, \big(\mathcal{A}^{(1)}(\mu) \,g^2\big)(t) & ~\le~ \sqrt{2} \, b^{-\frac{1}{2}},\nonumber\\
    g^{-2}(t)\, \big(\mathcal{A}^{(2)}(\mu) \,g^2\big)(t) & ~\le~ \frac{\mu^4}{2^4 \,3^2 \,b},\nonumber\\
    g^{-2}(t)\, \big( \mathcal{A}^{(3)}(\mu) \, g^2\big)(t) & ~\le~ \frac{3^{3/4}}{2^2\, \sqrt[4]{b}},\nonumber\\
   g^{-2}(t)\, \big( \mathcal{A}^{(4)}(\mu) \, g^2\big) (t) & ~\le~ \frac{\mu^4}{2^4 \, 3^{5/4}} \, b^{-3/4}.
\label{eq:gtotheminus2estimates}
\end{align}
Using \eqref{est:NormA1}, we can employ these inequalities (whose right hand sides are inversely proportional to powers of $b$) to estimate the norm of $A$. According to \eqref{est:NormA1}, we have, first on $\mathscr{D}$ and by linear extension also on the whole of $\Banach_g$: 
\be
	\|A\| ~\leq~ \|K\| \, \sup_{t>0}\, g^{-2}(t) \big((\id + 8\mathcal{A}(\mu)) \, g^2 \big)(t) .
\ee
Now we use \eqref{eq:gtotheminus2estimates} for the various contributions $A^{(k)}(\mu) $ to $\mathcal{A}(\mu) = \sum_{k=1}^4 A^{(k)}(\mu)$, finding:
\begin{align}
	\|A\| & ~\le~ \|K\| + \frac{2^{3.5} \|K\|}{b^{1/2}} + \frac{\mu^4 \|K\|}{18 b} + \frac{2\cdot 3^{3/4} \|K\|}{\sqrt[4]{b}} + \frac{\mu^4 \|K\|}{2 (3^5 b^3)^{1/4}}\nonumber\\
&\overset{b\ge 1}{\le}~ \|K\| + \frac{\|K\|}{\sqrt[4]{b}} \left(2^{3.5}+ \mu^4/18 + 2 \cdot 3^{3/4} + \mu^4/(2 \cdot 3^{5/4})\right)\\
&<~ \|K\| + \frac{\|K\|}{\sqrt[4]{b}} (16 + \mu^4 ).
\end{align}
Recalling that $b=\frac{\|K\|^4}{(1-\|K\|)^4} (16+\mu^4)^4$ (see \eqref{eq:defb}), we finally obtain that:
\be
	\|A\| ~<~ \|K\| + \frac{\|K\|}{b^{1/4}} (16 + \mu^4 ) ~= ~ \|K\| + 1-\|K\| =1.
\ee
We have thus shown that $A$ defines (by linear extension) a contraction on $\Banach_g$. Thus, the Neumann series $\psi = \sum_{k=0}^\infty A^k \psi^\free$ yields the unique (global-in-time) solution of the equation \(\psi=\psi^\free+A\psi\). \qed

\section{Conclusion and outlook} \label{sec:discussion}

Extending previous work for Klein-Gordon particles \cite{mtve,int_eq_curved} to the Dirac case, we have established the existence of dynamics for a class of integral equations which express direct interactions with time delay at the quantum level. To obtain this result, we have assumed a cutoff of the spacetime before $t=0$. It has been demonstrated that the Big Bang singularity can naturally provide such a cutoff. Remarkably, this yields a class of rigorous interacting models in 1+3 spacetime dimensions.

Compared to the previous works \cite{mtve,int_eq_curved}, our techniques have been modified and improved. Instead of demonstrating explicitly the convergence of the Neumann series by iterating the estimate \eqref{eq:spatialnormapsi} -- which is lengthy -- we have here succeeded in directly showing that $A$ is a contraction on the weighted space $\Banach_g$ for a suitable $g$. This also has the advantage that no arbitrary final time $T$ as in \cite{mtve,int_eq_curved} had to be introduced which could only later be taken to infinity by an additional argument (involving a change of Banach space).

The main challenge in our work has been the non-Markovian nature of the dynamics. This has made it necessary to directly prove global existence in time instead of concatenating short-time solutions on small time intervals (which would have been easier to obtain). Apart from this, the distributional derivatives in the Green's functions of the Dirac equation have made the analysis substantially more difficult than in the Klein-Gordon case. Compared to the latter, we have also treated the massive case (which was not considered in \cite{mtve} for 1+3 dimensions).

Our results are furthermore characterized as follows. We have shown that the wave function is determined by Cauchy data at the initial time (corresponding to the Big Bang singularity); however, no Cauchy problem is available at different times. The main requirement of our theorems is a smallness condition on the interaction kernel $K$, demanding that both $K$ and certain first and second order derivatives of $K$ must be bounded and not too large. This still admits a wide class of interaction kernels, and we emphasize that in no way the interaction needs to be small compared to the size of the domain of the wave function. The latter is a common requirement for Fredholm integral equations but it would make the result worthless for infinite spatio-temporal domains.

Besides, we have assumed that $K$ is complex-valued while it could be matrix-valued in the most general case. The reason for this assumption is that our proof requires the integral operator $A$ to be a map from a certain Sobolev space onto itself in which weak derivatives with respect to the Dirac operators of the two particles can be taken. If $K$ were matrix-valued, it would not commute with these Dirac operators in general. Then $A \psi$ would contain new types of weak derivatives which cannot be taken in the initial Sobolev space. As illustrated in Sec.\ \ref{sec:choiceofB}, this creates a situation where more and more derivatives have to be controlled, possibly up to infinite order where the success of an iteration scheme seems unlikely. At present, we do not know how to deal with this issue. 
Improving on this point, however, defines an important task for future research, as e.g.\ electromagnetic interactions involve interaction kernels proportional to $\gamma_1^\mu \gamma_{2\mu}$ (see \cite{direct_interaction_quantum}).

In addition, it would be desirable to generalize our work in the following regards.
\begin{itemize}
	\item \textit{$N$ particle integral equations.} Our hope is that our work could contribute to the formulation of a rigorous relativistic many-body theory that can be applied for finite times, not only for scattering processes. An important step in this direction is to treat an arbitrary fixed number $N \in \N$ of particles (setting aside particle creation and annihilation).  A class of possible $N$-particle integral equations has been suggested in \cite{direct_interaction_quantum}. It has the schematic form
	\be
		\psi(x_1,...,x_N) = \psi^\free + \sum_{i < j} \int d^4 x_i' \, d^4 x_j' S_i(x_i-x_i')\, S_j(x_j-x_j')\,K_{ij}(x_i',x_j')\, \psi(...x_i', ...,x_j',...).
	\ee
  It might well be possible to prove the existence and uniqueness of solutions for that equation using the methods developed in the present paper.
	\item \textit{Singular interaction kernels.} The physically most natural interaction kernel is given by a delta function along the light cone, $K(x_1,x_2) \propto \delta((x_1-x_2)_\mu (x_1-x_2)^\mu)$. Getting closer to this case is one of our central goals. Apart from approaching the problem head-on by suitably interpreting the distributional expressions and trying to prove the existence of solutions of the resulting singular integral equation, which seems difficult, one could also try to make smaller steps first. For example, one could decompose $\delta((x_1-x_2)_\mu (x_1-x_2)^\mu))$ into $\frac{1}{2 |\vx_1-\vx_2|} [ \delta(x_1^0-x_2^0 - |\vx_1-\vx_2|) + \delta(x_1^0-x_2^0 + |\vx_1-\vx_2|) ]$ and only then replace the delta functions with a peaked but smooth function, keeping the singular factor $1/|\vx_1-\vx_2|$. This has been done in \cite{mtve} for the Klein-Gordon case. In the Dirac case, the distributional derivatives make a generalization of that result difficult, and we have not attempted it here. However, it is conceivable that a suitable modification of our techniques could make it possible to treat this case.\\
Another interesting question is whether the smallness condition on $K$ can be alleviated such that arbitrarily peaked functions are admitted. This could allow taking a limit where $K$ approaches the delta function along the light cone.
\end{itemize}

\noindent \textbf{Acknowledgments.}\\[1mm]
We would like to thank Volker Bach, Dirk Deckert, Stefan Teufel and Roderich Tumulka for helpful discussions. 
M. N. acknowledges funding from Cusanuswerk and from the Elite Network of Bavaria, through the Junior Research Group `Interaction Between Light and Matter'. We thank an anonymous referee for particularly helpful suggestions.
\\[1mm]
\begin{minipage}{15mm}
\includegraphics[width=13mm]{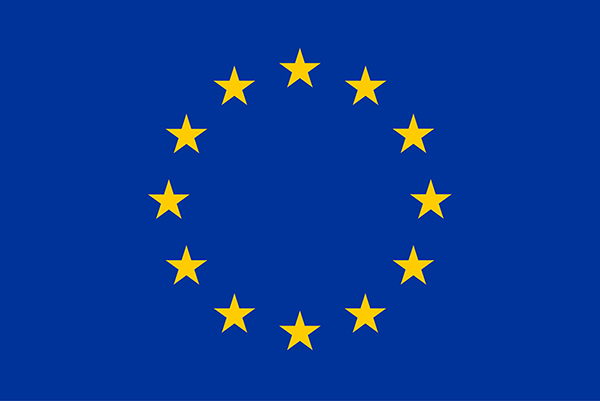}
\end{minipage}
\begin{minipage}{143mm}
This project has received funding from the European Union's Framework for
Re- \\ search and Innovation Horizon 2020 (2014--2020) under the Marie Sk{\l}odowska-
\end{minipage}\\[1mm]
Curie Grant Agreement No.~705295.


\begin{thebibliography}{10}

\bibitem{dirk_martin_2018}
D.-A. {Deckert} and M.~{Oelker}.
\newblock {Distinguished self-adjoint extension of the two-body Dirac operator
  with Coulomb interaction}.
		\newblock {\em Ann. Henri Poincaré: \url{https://doi.org/10.1007/s00023-019-00802-6}20:2407 (2019)} 

\bibitem{thirring_model}
W.~Thirring.
\newblock {A Soluble Relativistic Field Theory?}
\newblock {\em Ann. Phys.}, 3:91--112, 1958.

\bibitem{glimm_jaffe}
J.~{Glimm} and A.~{Jaffe}.
\newblock {\em {Quantum Physics -- A Functional Integral Point of View}}.
\newblock Springer, 1987.

\bibitem{jaffe_cft}
A.~{Jaffe}.
\newblock {Constructive Quantum Field Theory}.
\newblock \url{https://www.arthurjaffe.com/Assets/pdf/CQFT.pdf}.

\bibitem{dirac_32}
P.~A.~M. Dirac.
\newblock {Relativistic Quantum Mechanics}.
\newblock {\em Proc. R. Soc. Lond. A}, 136:453--464, 1932.

\bibitem{tomonaga}
S.~Tomonaga.
\newblock {On a Relativistically Invariant Formulation of the Quantum Theory of
  Wave Fields}.
\newblock {\em Prog. Theor. Phys.}, 1:27--42, 1946.

\bibitem{schwinger}
J.~Schwinger.
\newblock {Quantum Electrodynamics. I. A Covariant Formulation}.
\newblock {\em Phys. Rev.}, 74(2162):1439--1461, 1948.

\bibitem{guenther_1952}
M.~{G\"unther}.
\newblock {Many-Times Formalism and Coulomb Interaction}.
\newblock {\em Phys. Rev.}, 88(6):1411--1421, 1952.

\bibitem{marx_1974}
E.~Marx.
\newblock {Many-Times Formalism and Coulomb Interaction}.
\newblock {\em Int. J. of Theor. Phys.}, 9(3):195--217, 1974.

\bibitem{schweber}
S.~Schweber.
\newblock {\em {An Introduction to Relativistic Quantum Field Theory}}.
\newblock Dover, 2005.
\newblock Originally published in 1961.

\bibitem{drozvincent_1981}
Ph. {Droz-Vincent}.
\newblock {Relativistic Wave Equations for a System of Two Particles with Spin
  1/2}.
\newblock {\em Lettere al Nuovo Cimento}, 30:375--378, 1981.

\bibitem{sazdjian_2bd}
H.~Sazdjian.
\newblock {Relativistic wave equations for the dynamics of two interacting
  particles}.
\newblock {\em Phys. Rev. D}, 33:3401--3424, 1986.

\bibitem{2bdem}
H.~W. {Crater} and P.~{Van Alstine}.
\newblock {A tale of three equations: Breit, Eddington-Gaunt, and Two-Body
  Dirac}.
\newblock {\em Found. Phys.}, 27:67--79, 1997.

\bibitem{nogo_potentials}
S.~{Petrat} and R.~{Tumulka}.
\newblock {Multi-Time Schr\"odinger Equations Cannot Contain Interaction
  Potentials}.
\newblock {\em J. Math. Phys.}, 55(032302), 2014.
\newblock \url{https://arxiv.org/abs/1308.1065}.

\bibitem{qftmultitime}
S.~{Petrat} and R.~{Tumulka}.
\newblock {Multi-Time Wave Functions for Quantum Field Theory}.
\newblock {\em Ann. Phys.}, 345:17--54, 2014.
\newblock \url{https://arxiv.org/abs/1309.0802v3}.

\bibitem{multitime_pair_creation}
S.~{Petrat} and R.~{Tumulka}.
\newblock {Multi-time formulation of pair creation}.
\newblock {\em J. Phys. A: Math. Theor.}, 47(11):112001, 2014.

\bibitem{1d_model}
M.~Lienert.
\newblock {A relativistically interacting exactly solvable multi-time model for
  two massless Dirac particles in 1+1 dimensions}.
\newblock {\em J. Math. Phys.}, 56(4):042301, 2015.
\newblock \url{https://arxiv.org/abs/1411.2833}.

\bibitem{nt_model}
M.~{Lienert} and L.~{Nickel}.
\newblock {A simple explicitly solvable interacting relativistic N-particle
  model}.
\newblock {\em J. Phys. A: Math. Theor.}, 48(32):325301, 2015.
\newblock \url{https://arxiv.org/abs/1502.00917}.

\bibitem{2bd_current_cons}
M.~Lienert.
\newblock {On the question of current conservation for the Two-Body Dirac
  equations of constraint theory}.
\newblock {\em J. Phys. A: Math. Theor.}, 48(32):325302, 2015.
\newblock \url{https://arxiv.org/abs/1501.07027}.

\bibitem{deckert_nickel_2016}
D.-A. {Deckert} and L.~{Nickel}.
\newblock {Consistency of multi-time Dirac equations with general interaction
  potentials}.
\newblock {\em J. Math. Phys.}, 57(7):072301, 2016.
\newblock \url{https://arxiv.org/abs/1603.02538}.

\bibitem{lpt_2017b}
M.~{Lienert}, S.~{Petrat}, and R.~{Tumulka}.
\newblock {Multi-Time Wave Functions Versus Multiple Timelike Dimensions}.
\newblock {\em Found. Phys.}, 47:1582--1590, Oct 2017.
\newblock \url{https://arxiv.org/abs/1708.03376}.

\bibitem{generalized_born}
M.~{Lienert} and R.~{Tumulka}.
\newblock {Born's Rule for Arbitrary Cauchy Surfaces}, 2017.
\newblock \url{https://arxiv.org/abs/1706.07074}.

\bibitem{ibc_model}
M.~{Lienert} and L.~{Nickel}.
\newblock {Multi-time formulation of creation and annihilation of particles via
  interior-boundary conditions}.
\newblock Rev. Math. Phys.: \url{https://arxiv.org/abs/1808.04192} 2:2050004, 2020.

\bibitem{phd_nickel}
L.~Nickel.
\newblock {PhD thesis. On the Dynamics of Multi-Time Systems}, 2019.
\newblock Mathematical Institute, Ludwig-Maximilians-Universit\"at, Munich,
  Germany.

\bibitem{dice_paper}
M.~{Lienert}, S.~{Petrat}, and R.~{Tumulka}.
\newblock {Multi-time wave functions}.
\newblock {\em J. Phys. Conf. Ser.}, 880(1):012006, 2017.
\newblock \url{https://arxiv.org/abs/1702.05282}.

\bibitem{direct_interaction_quantum}
M.~{Lienert}.
\newblock {Direct interaction along light cones at the quantum level}.
\newblock {\em J. Phys. A: Math. Theor.}, 51(43):435302, 2018.
\newblock \url{https://arxiv.org/abs/1801.00060}.

\bibitem{wf1}
J.~A. Wheeler and R.~P. Feynman.
\newblock {Interaction with the Absorber as the Mechanism of Radiation}.
\newblock {\em Rev. Mod. Phys.}, 17:157--181, 1945.

\bibitem{wf2}
J.~A. Wheeler and R.~P. Feynman.
\newblock {Classical Electrodynamics in Terms of Direct Interparticle Action}.
\newblock {\em Rev. Mod. Phys.}, 21:425--433, 1949.

\bibitem{mtve}
M.~{Lienert} and R.~{Tumulka}.
\newblock {A new class of Volterra-type integral equations from relativistic
  quantum physics}.
\newblock \textit{J. Integral Equations Applications} 31.4(2019): 535-569.
  \url{https://projecteuclid.org/euclid.jiea/1536804036} (2019).

\bibitem{int_eq_curved}
M.~{Lienert} and R.~{Tumulka}.
\newblock {Interacting relativistic quantum dynamics of two particles on
  spacetimes with a Big Bang singularity}.
\newblock \textit{J. Math. Phys.} 60:4(2019): 042302.
  \url{https://arxiv.org/abs/1805.06348}.

\bibitem{bs_equation}
E.~E. Salpeter and H.~A. Bethe.
\newblock {A Relativistic Equation for Bound-State Problems}.
\newblock {\em Phys. Rev.}, 84:1232--1242, 1951.

\bibitem{dubinskii_1991}
Yu.~A. {Dubinskii}.
\newblock {Sobolev spaces of infinite order}.
\newblock {\em Russ. Math. Surv.}, 46(6):107--147, 1991.

\bibitem{ibison}
M.~{Ibison}.
\newblock {On the conformal forms of the Robertson-Walker metric}.
\newblock {\em J. Math. Phys.}, 48:122501, 2007.
\newblock \url{https://arxiv.org/abs/0704.2788}.

\bibitem{penrose_rindler}
R.~{Penrose} and W.~{Rindler}.
\newblock {\em {Spinors and Space-time, Volume 1}}.
\newblock Cambridge University Press, 1984.

\bibitem{haantjes}
J.~{Haantjes}.
\newblock {The conformal Dirac equation}.
\newblock {\em Proc. Kon. Nederl. Acad. Wetensch.}, 44:324--332, 1941.

\bibitem{john}
R.~W. {John}.
\newblock {The Hadamard Construction of Green's Functions on a Curved
  Space-time with Symmetries}.
\newblock {\em Ann. Phys.}, 48(7):531--544, 1987.

\end{thebibliography}

\end{document}